\documentclass[10pt, conference]{IEEEtran}
\IEEEoverridecommandlockouts
\usepackage{cite}
\usepackage{amsmath,amssymb,amsfonts,amsthm}
\usepackage{algorithmic}
\usepackage{graphicx}
\usepackage{textcomp}
\usepackage{xcolor}
\def\BibTeX{{\rm B\kern-.05em{\sc i\kern-.025em b}\kern-.08em
    T\kern-.1667em\lower.7ex\hbox{E}\kern-.125emX}}

\begin{document}

\newtheorem{thm}{{\bf Theorem}}
\newtheorem{pro}{\bf Proposition}
\newtheorem{corollary}{\bf Corollary}
\newtheorem{lem}{{\bf Lemma}}
\newcommand{\xd}{x^{\texttt{ded}}}
\newcommand{\xs}{x^{\texttt{sh}}}

\newcommand{\yd}{y^{\texttt{ded}}}
\newcommand{\ys}{y^{\texttt{sh}}}

\newcommand{\kt}{{(k, t)}}

\newcommand{\cals}{\mathcal{S}}
\newcommand{\calu}{\mathcal{U}}

\newcommand{\rev}[1]{  
	{\textcolor{red}{#1}}{}} 
\newcommand{\bo}[1]{  
	{\textcolor{blue}{(Bo says:  #1)}}{}} 

\title{HyRA: A Hybrid Resource Allocation Framework for RAN Slicing}

\author{
\IEEEauthorblockN{Mohammad Zangooei\IEEEauthorrefmark{1}, 
Bo Sun\IEEEauthorrefmark{2}, 
Noura Limam\IEEEauthorrefmark{1}, 
Raouf Boutaba\IEEEauthorrefmark{1}}
\IEEEauthorblockA{\IEEEauthorrefmark{1}University of Waterloo, 
Email: \{mzangooei, noura.limam, rboutaba\}@uwaterloo.ca}
\IEEEauthorblockA{\IEEEauthorrefmark{2}University of Ottawa, 
Email: bo.sun@uottawa.ca}
}

\maketitle

\begin{abstract}
The advent of 5G and the emergence of 6G networks demand unprecedented flexibility and efficiency in Radio Access Network (RAN) resource management to satisfy diverse service-level agreements (SLAs). Existing RAN slicing frameworks predominantly rely on per-slice resource reservation, which ensures performance isolation but leads to inefficient utilization, particularly under bursty traffic. We introduce \textbf{HyRA}, a hybrid resource allocation framework for RAN slicing that combines dedicated per-slice allocations with shared resource pooling across slices. HyRA preserves performance isolation while improving resource efficiency by leveraging multiplexing gains in bursty traffic conditions. We formulate this design as a bi-level stochastic optimization problem, where the \textit{outer loop} determines the dedicated and shared resource budgets and the \textit{inner loop} performs per-UE scheduling under a novel water-filling approach. By using the sample-average approximation, the Karush–Kuhn-Tucker (KKT) conditions of the inner loop, and Big-M encoding, we transform the problem into a tractable mixed-integer program that standard optimization solvers can solve. Extensive simulations under diverse demand patterns, SLA configurations, and traffic burstiness show that HyRA achieves up to $\mathbf{50–75\%}$ spectrum savings compared to dedicated-only and shared-only baselines. These results highlight HyRA as a viable approach for resource-efficient, SLA-compliant RAN slicing in future mobile networks.
\end{abstract}

\begin{IEEEkeywords}
Mobile Networks, RAN Slicing, Resource Management, Optimization
\end{IEEEkeywords}

\section{Introduction}

Mobile networks have become a cornerstone of modern digital infrastructure, enabling widespread services such as conversational streaming, high-definition video streaming, and real-time gaming \cite{jeffrey14what}. Looking ahead, they are expected to support a new generation of mission-critical applications, including autonomous driving, remote surgery, industrial automation, and massive-scale Internet of Things (IoT) \cite{afolabi2018network}. Meeting these demands, especially with the advent of 5G and the promise of 6G, requires a fundamental rethinking of how networks are designed, operated, and optimized.  

Network slicing enables the dynamic creation of virtual networks tailored to tenant-specific service-level agreements (SLAs) \cite{afolabi2018network}. From an operator’s perspective, optimizing slice resource consumption is crucial for reducing operational expenditures. Yet, this optimization remains inherently challenging, particularly for Radio Access Network (RAN) slicing due to the stochastic nature of wireless transmission, the variability of user demand, and the heterogeneity of SLAs.

Most existing RAN slicing frameworks adopt a dedicated-resource approach, in which radio resources are pre-assigned to individual slices and periodically adjusted to accommodate changing network conditions~\cite{zangooei2023flexible, polese2022colo, liu2020constrained, sulaiman2023generalizable, alcaraz2022model, garces2017network, zheng2018statistical, balasingam2024application}. While this approach ensures strict SLA compliance, it neglects the potential for resource sharing among slices that arises from the inherently bursty and time-varying nature of real-world cellular data traffic. Empirical studies have shown that, even when aggregated over many users, traffic exhibits pronounced spikes and idle periods across multiple time scales~\cite{perez2023characterizing}. In such conditions, allocating resources based on peak demand results in substantial over-provisioning and disregards multiplexing opportunities across slices. The limitations of prior works can be partly attributed to their use of overly simplistic traffic models, such as constant bit rate or Poisson arrivals~\cite{zangooei2023flexible, polese2022colo, sulaiman2023generalizable, liu2020constrained, alcaraz2022model, garces2017network, zheng2018statistical, balasingam2024application}, which fail to capture the bursty nature of real-world traffic and consequently misrepresent the trade-off between resource efficiency and SLA satisfaction~\cite{perez2023characterizing}.

Two common paradigms for RAN resource allocation are \emph{dedicated resource reservation}, which offers strong isolation and guarantees but suffers from low utilization, and \emph{shared resource pooling}, which enables statistical multiplexing gains but provides limited SLA assurances. In this paper, we introduce \textbf{HyRA}, a hybrid resource allocation framework that combines the strengths of both paradigms: dedicated allocations guarantee consistent performance for each slice, while a shared resource pool flexibly accommodates bursty and time-varying traffic demands. By combining isolation with multiplexing, HyRA accounts for traffic bursts and aims to reduce overall spectrum consumption while maintaining SLA compliance.

We formulate the dynamics of this design as a bi-level stochastic optimization problem. The bi-level structure reflects hierarchical decision-making: the \emph{outer loop} determines resource budgets by deciding how many physical resource blocks (PRBs) to reserve and share. In contrast, the \emph{inner loop} captures the real-time scheduling of PRBs to User Equipments (UEs) under these budgets. The outer loop seeks the minimum PRB usage that satisfies the slice SLA constraints, while the inner loop aims to balance spectrum efficiency and fairness using a novel water-filling approach. This model mirrors control loops in state-of-the-art O-RAN architectures, where the slicing xApp sets slice policies in near real time, and the gNB scheduler executes per-UE allocations in real time~\cite{ORAN2022}. 

A key challenge in this formulation lies in the stochastic nature of traffic and channel conditions, which are unknown at the time of the outer loop decision. To handle this uncertainty, we adopt a sample-average approximation (SAA) framework that generates multiple realizations of traffic and channel states to approximate the expected service performance. The resulting bi-level optimization problem is then transformed into a single-level problem by embedding the Karush–Kuhn–Tucker (KKT) conditions of the inner loop into the outer loop. Next, we convert the fractional constraints into an equivalent set of quadratic constraints and apply Big-M encoding to linearize them. This leads to a mixed-integer program (MIP) that can be efficiently solved using general-purpose solvers such as \texttt{Gurobi} or \texttt{CPLEX}.

We evaluate HyRA's performance under diverse network conditions, including varying traffic demands, SLA definitions, and slice counts. The results consistently demonstrate that HyRA achieves substantial spectrum savings while maintaining SLA compliance across all scenarios. Under per-UE SLA enforcement, HyRA reduces PRB allocation by approximately $\mathbf{65–75\%}$ compared to the average of dedicated- and shared-only baselines. When the SLA is relaxed to slice-level aggregation, HyRA continues to outperform both baselines with $\mathbf{55–65\%}$ savings. Furthermore, scalability experiments show that, as the number of slices increases, HyRA maintains its efficiency advantage, reducing PRB consumption by up to $\mathbf{60–70\%}$ relative to the baselines' average resource consumption. Finally, we assess HyRA's robustness across varying levels of traffic burstiness. The results show that as burstiness increases, HyRA's relative efficiency further improves. Overall, the results confirm that HyRA effectively combines the isolation benefits of dedicated provisioning with the statistical multiplexing gains of shared pooling.

With this work, we demonstrate the untapped potential of jointly leveraging shared and dedicated resource pools across network slices by solving the RAN slicing problem optimally. While such an optimal solution provides valuable insights into the theoretical limits of performance, it is, by nature, computationally intractable for real-time operation. This observation aligns with the broader body of research in the field, where heuristic and learning-augmented approaches~\cite{zangooei2023flexible, polese2022colo, liu2020constrained, sulaiman2023generalizable, liu2021onslicing, alcaraz2022model, garces2017network, zheng2018statistical} have emerged to balance solution optimality with time complexity. Building upon the findings of this study, we plan to develop a practical algorithm that leverages the benefits of both sharing and dedication, while maintaining manageable computational complexity for real-world deployment.

\section{Related Work}
RAN slicing has garnered extensive attention from the research community, with numerous methodologies proposed, including heuristic-based solutions~\cite{chen2023channel, d2019slice, ksentini2017toward}, reinforcement learning~\cite{zangooei2023flexible, polese2022colo, liu2020constrained, sulaiman2023generalizable}, online learning~\cite{liu2021onslicing, alcaraz2022model}, and game-theoretic formulations~\cite{garces2017network, zheng2018statistical}. Most of these studies rely on \emph{hard slicing}, where resources allocated to each slice are strictly isolated and cannot be reused by other slices. These methods typically aim to minimize the total amount of resources dedicated across all slices while still meeting their respective service-level guarantees. However, such designs inherently restrict multiplexing opportunities among slices since idle resources in lightly loaded slices cannot be exploited by others, and thus often result in suboptimal spectrum utilization.

To mitigate these inefficiencies, recent works~\cite{zhao2025adaslicing, yang2024advancing} have introduced so-called \emph{soft slicing} schemes, in which the unused portion of a slice’s dedicated resources can be temporarily shared with other slices. While this approach improves utilization, the sharing process is typically handled in a reactive manner rather than through joint optimization. To be specific, existing soft-slicing approaches typically exhaust all resources dedicated to a slice to immediately serve its queued traffic, sharing only the remaining unused portion with other slices. However, such a service strategy overlooks the fact that certain traffic flows can tolerate short waiting times within their allowable delay budgets, rather than being served instantaneously. By neglecting this temporal flexibility, current soft-slicing designs fail to exploit the full multiplexing gains that could substantially reduce the system’s resource demand.

We argue that genuine soft slicing is achieved when only the minimum resources are dedicated per slice, while the remaining capacity is managed within a common shared pool. This design enables adaptive multiplexing across slices according to their instantaneous traffic demands and delay budgets. In this context, our work departs from prior post hoc sharing schemes by jointly optimizing dedicated and shared resource allocations within a unified framework. 

\section{System Model}
\subsection{Traffic Model} \label{sec:traffic_model}
Empirical measurements of Internet and mobile traffic demonstrated that packet flows are highly \emph{bursty} even in millisecond time scale and exhibit strong correlations across multiple time scales~\cite{park2000self, perez2023characterizing}. 
This burstiness implies that exponential or uniform distributions cannot adequately represent traffic, as such models underestimate the likelihood of extreme events such as sudden traffic surges.

To synthesize trustworthy traces that reproduce this statistical behavior, we employ the \emph{Pareto distribution} for both inter-arrival times and packet sizes. The Pareto distribution is widely recognized as a canonical model for heavy-tailed processes, where a small number of large values (long inter-arrival gaps or large packets) occur with non-negligible probability \cite{park2000self, leland2002self,  sharafzadeh2023understanding}. This property makes it well-suited to represent the long-range dependence and self-similarity observed in real-world traffic. This heavy-tailed behavior manifests in two ways:
\begin{itemize}
    \item Pareto-distributed packet sizes capture the empirical fact that a small fraction of very large packets contributes disproportionately to overall load.
    \item Pareto-distributed inter-arrival times induce traffic correlations across multiple scales. When the tail index $\alpha$ lies in the range $1 \!<\! \alpha \!<\! 2$, the aggregation of many such flows yields a self-similar process with long-range dependency~\cite{taqqu1997proof}.
\end{itemize}
In fact, the tail parameter $\alpha$ governs the severity of burstiness:
\begin{itemize}
    \item For $0 \! < \!\alpha\! \leq \! 1$, the mean of the distribution is infinite, leading to unrealistic traffic load;
    \item For $1 \! < \!\alpha\! \leq \! 2$, the mean is finite but the variance is infinite, reflecting realistic heavy-tailed burstiness;
    \item For $\alpha \!>\! 2$, both mean and variance are finite, reducing burstiness to light-tailed behavior.
\end{itemize}
In practice, selecting $\alpha$ in the range $(1,2]$ captures the statistical properties of real-world traffic: the average packet size and inter-arrival time remain bounded, while large fluctuations are sufficiently frequent.

\subsection{Channel Model}
In wireless communications, the received signal quality fluctuates due to a combination of large-scale and small-scale propagation effects. Large-scale effects include path loss, which captures the attenuation of the signal with distance from the transmitter, and shadowing, which accounts for slow variations caused by obstacles such as buildings. Small-scale effects, on the other hand, arise from multipath fading, where multiple delayed signal replicas interfere constructively and destructively at the receiver. Together, these phenomena shape the instantaneous signal-to-noise ratio (SNR) a UE experiences. This SNR directly dictates the number of bits that each PRB, which is the fundamental unit of radio resource allocation consisting of a fixed bandwidth in the frequency domain over one time slot, can reliably carry~\cite{goldsmith2005wireless}.

In practice, UEs quantize the instantaneous SNR they experience and report this quantized value, known as the Channel Quality Indicator (CQI), to the gNB. Each CQI value corresponds to a predefined modulation and coding scheme (MCS), which determines the spectral efficiency (SE) $\eta$ measured in information bits per modulation symbol. The gNB uses the reported CQI to select the MCS that maximizes throughput while maintaining a target block error rate (typically $10\%$). Once the MCS is fixed, the number of information bits per PRB is obtained by multiplying $N_{\texttt{sym}}$, the number of modulation symbols contained in a PRB (12 subcarriers × 14 OFDM symbols = 168 symbols per slot), by the corresponding SE: $N_{\texttt{bits/PRB}} = N_{\texttt{sym}} \times \eta$~\cite{3gpp38.214}.

To emulate realistic wireless channel conditions, we rely on link-level traces generated with the standardized fading channel models defined in TS 36.104 \cite{EUTRAfading}, namely Extended Pedestrian A (EPA), Extended Vehicular A (EVA), and Extended Typical Urban (ETU), using the ns-3 simulator~\cite{alcaraz2022model}. These models capture both large-scale and small-scale fading dynamics under different mobility scenarios: EPA for low-mobility pedestrian environments, EVA for medium-to-high mobility vehicular settings, and ETU for dense urban multipath conditions. The resulting traces produce stochastic variations in the instantaneous SNR experienced by a UE. From these traces, we compute the corresponding spectral efficiency $\eta$ values on a per–time slot basis, which are then used as inputs to our formulations in the following sections.

\subsection{Slicing Model}
In our system, the mobile network operator (MNO) supports multiple network slices, each tailored to a specific service type and corresponding quality of service (QoS) requirements. Each slice serves a group of UEs with similar delay budgets. These heterogeneous delay budgets across slices require the network to differentiate its service, as discussed in 3GPP TS 23.501~\cite{3gpp2024sa5gs} and studies such as \cite{polese2022colo, liu2020constrained, garces2017network}.
For instance, ultra-reliable low-latency communication (uRLLC) services, such as autonomous driving or industrial control, require extremely short delay budgets on the order of a few milliseconds, while enhanced mobile broadband (eMBB) services can tolerate more relaxed latency targets around tens of milliseconds. Also, the delay budget implicitly captures the high-throughput requirements of eMBB slices, where applications such as video streaming or augmented reality generate large traffic volumes. This large traffic volume translates into higher queueing delays when resources are limited, so ensuring that the delay remains within a given budget inherently requires allocating sufficient resources to sustain the necessary throughput.

The MNO is responsible for allocating limited PRBs to meet slices' delay requirements. The TS 28.541 resource model architecture enables operators to define, for each slice instance, whether it uses a dedicated or a shared resource pool~\cite{3gpp2021nrm}. Allocating more resources to stringent slices, such as uRLLC, ensures that latency targets are met but increases overall resource consumption. Conversely, relaxing the allocation to less-demanding slices reduces total cost or frees up capacity to serve additional UEs. Therefore, the central challenge is to balance resource efficiency and service guarantees: ensuring that each slice meets its delay budget while minimizing total resource consumption. By optimizing this trade-off, the operator can either reduce operational costs for the same user base or, under a fixed total resource budget, increase the number of UEs that can be simultaneously supported. This forms the foundation of our slicing model, which explicitly captures service heterogeneity and the impact of resource allocation decisions on delay performance across different slices.

Our RAN slicing model aligns with O-RAN architecture~\cite{ORAN2022} where an xApp deployed on a near-real-time radio intelligent controller periodically determines slicing policy and communicates it to the gNB (the policy is detailed in the next section). The gNB’s scheduler then allocates PRBs to individual UEs in real time, while ensuring that the allocation adheres to the most recent slicing policy.

\section{HyRA: A Hybrid Resource Allocation Framework for RAN Slicing} \label{sec:prob_form}
We consider a RAN slicing policy that specifies (i) the number of PRBs reserved as \emph{dedicated} to each slice, and (ii) the pool of PRBs \emph{shared} among all slices. Because future channel conditions and traffic arrivals are unknown at the time of decision, the xApp must solve a \emph{stochastic optimization problem} that balances two objectives: minimizing resource usage and satisfying slice-level SLAs.

We formulate this as a \emph{bi-level stochastic optimization problem} over a time window of length $T$~ms. The \emph{outer loop} optimization problem determines the budgets of dedicated and shared PRBs for each slice, while the \emph{inner loop} optimization problem allocates PRBs to UEs within those budgets. Let $\mathcal{S} = \{1, \dots, S\}$ denote the set of slices, and let $\mathcal{U}_s$ be the set of UEs in slice~$s$. The overall set of UEs in the system is given by $\mathcal{U} = \bigcup_{s \in \mathcal{S}} \mathcal{U}_s$.

\subsection{Outer Loop Optimization Problem} \label{sec:outer}
At the beginning of each time window, the outer loop optimization problem determines how many PRBs are reserved as \emph{dedicated} resources to each slice ($\xd_s$) and how many are pooled as \emph{shared} resources across all slices ($\xs$). 
Let $\boldsymbol{x} = \{\{\xd_s\}_{s\in\cals}, \xs \}$ denote the outer loop decision variable.
Based on fine-grained scheduling of these PRBs, each UE~$i$ in slice~$s$ experiences an average delay $d_i$ over the time window. The objective of the outer loop is to minimize the total allocated resources, 
\vspace{-4pt}
\begin{equation}
\xs + \sum_{s=1}^{S} \xd_s,    
\end{equation}
while ensuring that UE $i$’s average delay $d_i$ remains below the SLA threshold $D_{s,i}$ of its corresponding slice. 

Since future traffic and channel conditions are stochastic, the delay cannot be computed deterministically. Instead, we adopt the sample-average approximation (SAA), a standard method that reformulates stochastic optimization as a deterministic problem over finite samples \cite{kleywegt2002sample}. Specifically, we assume a fixed set of UEs and generate $K$ independent time-series samples, each of length $T$, for the packet arrivals in bits 
\begin{equation}
\{A\kt: 1 \leq k \leq K,\; 1 \leq t \leq T\},    
\end{equation}
and for the SE in bits per resource element 
\begin{equation}
\{\eta\kt: 1 \leq k \leq K,\; 1 \leq t \leq T\},    
\end{equation}
representing channel quality. At each time step $t$ of sample $k$, the lower-level problem schedules PRBs to UE~$i$ from both dedicated and shared pools, denoted by $\yd_i\kt$ and $\ys_i\kt$. 
Let $\boldsymbol{y}(k,t) = \{\yd_i\kt, \ys_i\kt\}_{i\in\calu}$ denote the inner loop decision variable for sample $k$ and time $t$.

The queue dynamics of each UE evolve according to:
\begin{subequations}
\begin{align}
& Q_i(k,\!t\!+\!1) = \max \{ Q_i\kt + A_i\kt - S_i\kt, 0\} , \\
& S_i\kt \leq N_{\texttt{sym}} \, \eta_i\kt\,\big(\yd_i\kt + \ys_i\kt\big).
\end{align}
\end{subequations}
where $S_i\kt$ is the service in bits provided to UE~$i$.

By Little’s Law, UE~$i$'s experienced delay in sample~$k$ is
\begin{equation}
\hat{d}_i(k) = \frac{\sum_{t=1}^{T} Q_i\kt}{\sum_{t=1}^{T} A_i\kt}.
\end{equation}
Intuitively, the queue length represents backlog, and Little’s Law translates backlog into average delay. Then SAA approximates the expected delay of UE~$i$ as:
\vspace{-4pt}
\begin{equation}
\hat{d}_i = \frac{1}{K} \sum_{k=1}^{K} \hat{d}_i(k).
\end{equation}
\vspace{-4pt}

Finally, the outer loop optimization problem can be expressed as:
\vspace{-4.5pt}
\begin{subequations}
\begin{align}
\min & \; \xs + \sum_{s=1}^{S} \xd_s \\
\text{s.t.}
& \frac{1}{K} \sum_{k=1}^{K} \frac{\sum_{t=1}^{T} Q_i\kt}{\sum_{t=1}^{T} A_i\kt}
\leq D_{s,i}, & \forall i \label{eq:HyRA:sla}\\
& Q_i\kt + A_i\kt - S_i\kt \leq Q_i(k,\!t\!+\!1), & \forall i, k, t \label{eq:HyRA:queue1}\\
& S_i\kt \leq N_{\texttt{sym}}\, \eta_i\kt\,\big(\yd_i\kt + \ys_i\kt\big), & \forall i, k, t \label{eq:HyRA:queue2} \\
& 0 \leq Q_i\kt, & \forall i, k, t \\
& 0 \leq \xs,\;\; 0 \leq \xd_s, & \forall s\\
& \boldsymbol{y}\kt \in \mathcal{Y}(\boldsymbol{x};k,t).
\end{align}
\end{subequations}
where $\mathcal{Y}(\boldsymbol{x};k,t)$ is the inner loop solution set for given outer loop solution $\boldsymbol{x}$, sample $k$ at time step $t$.

This outer loop formulation captures the trade-off between minimizing spectrum usage and guaranteeing SLA.

\subsection{Inner Loop Optimization Problem}

The inner loop optimization problem determines how the PRBs allocated by the outer loop are distributed among UEs at each sample $k$ and time slot $t$. The scheduler must respect the dedicated and shared budgets determined in the outer loop while ensuring efficiency and fairness among UEs.
A natural design choice would be to maximize the total throughput by assigning resources to UEs with the best instantaneous channel conditions. However, such a \emph{max-throughput scheduler} risks starving weaker UEs. Conversely, a \emph{round-robin scheduler} that ignores channel heterogeneity guarantees fairness but leads to inefficient spectrum utilization. To balance these extremes, we adopt a \emph{utility-based formulation} that maximizes the sum of the concave utilities of UE rates. In particular, the logarithmic utility captures the principle of \emph{proportional fairness}, a well-established scheduling criterion in wireless networks \cite{kelly1998rate}. This approach prioritizes UEs with favorable channels while ensuring that every UE receives a non-zero share of resources.

The inner loop problem for sample $\kt$ is defined below.
\begin{subequations}
\begin{align}
\max & \sum_{i \in \mathcal{U}} 
\log \!\left( 1 \!+\! \eta_i\kt\,\big[\yd_i\kt \!+\! \ys_i\kt\big] \right) \label{eq:inner:obj} \\
\text{s.t.}
& \sum_{i \in \mathcal{U}_s} \yd_i\kt \leq \xd_s, & \forall s \label{eq:inner:c_ded} \\
& \sum_{i \in \mathcal{U}} \ys_i\kt \leq \xs, & \label{eq:inner:c_sh} \\
& 0 \leq \yd_i\kt,\;\; 0 \leq \ys_i\kt. & \forall i \label{eq:inner:c_pos}
\end{align}
\end{subequations}

Notably, the resource units should be integer-valued, as they correspond to the count of PRBs. However, to alleviate the combinatorial complexity of the problem, we relax this constraint and treat the variables $\boldsymbol{x}$ and $\boldsymbol{y}\kt$ as continuous during optimization. After solving the problem, we apply a ceiling operation to obtain legitimate integer-valued PRB allocations, with the resulting cost typically negligible due to the fine granularity of PRBs.

\subsection{Two-stage Water-filling Scheduler}
The optimization problem in the inner loop~(\ref{eq:inner:obj}-\ref{eq:inner:c_pos}) can be intuitively understood through the lens of a two-stage \emph{water-filling} procedure. In classical water-filling~\cite{he2013water}, resources are allocated in proportion to channel quality, with the metaphor of pouring water into vessels of different base heights until they are leveled. Our formulation extends this intuition into two stages: dedicated leveling per slice, followed by shared leveling across slices.

To formally characterize the optimal inner loop allocation at instance $\kt$, we consider its Karush–Kuhn–Tucker (KKT) conditions~\cite{boyd2004convex}. Because the objective is concave and all constraints are affine, the KKT conditions are necessary and sufficient for optimality. We associate the following dual variables with instance $\kt$:
\begin{subequations}
\begin{align}
&\lambda_s\kt \geq 0, \forall s: \text{associated with constraint (\ref{eq:inner:c_ded})} \label{eq:dual_lambda} \\
&\nu\kt \geq 0: \text{associated with constraint (\ref{eq:inner:c_sh})} \label{eq:dual_nu} \\
&\gamma_i\kt \geq 0, \forall i: \text{associated with } \yd_i\kt \ge 0 \label{eq:dual_gamma} \\
&\sigma_i\kt \geq 0, \forall i: \text{associated with } \ys_i\kt \ge 0 \label{eq:dual_sigma}
\end{align}    
\end{subequations}

The KKT conditions then consist of:

\begin{itemize}
\item \textbf{Stationarity}
\begin{subequations}
\begin{align}
&\frac{\eta_i\kt}{1 + \eta_i\kt \left( \yd_i\kt + \ys_i\kt \right)} \notag \\ 
& \quad = \lambda_{s}\kt - \gamma_i\kt & \forall i\in\calu_s, s \label{eq:kkt:stat_d}\\ 
& \frac{\eta_i\kt}{1 + \eta_i\kt \left( \yd_i\kt + \ys_i\kt \right)} \notag \\
& \quad = \nu\kt - \sigma_i\kt, & \forall i\in\calu \label{eq:kkt:stat_s}
\end{align}
\end{subequations}

\item \textbf{Complementary slackness}
\begin{subequations}
\begin{align}
&\lambda_s\kt \left( \xd_s - \sum_{i \in \mathcal{U}_s} \yd_i\kt  \right) = 0 & \forall s \label{eq:kkt:cs_d} \\
&\nu\kt\left( \xs - \sum_{i \in \mathcal{U}} \ys_i\kt\right) = 0 \label{eq:kkt:cs_s} \\
&\gamma_i\kt \cdot \yd_i\kt = 0 & \forall i \label{eq:kkt:cs_id} \\
&\sigma_i\kt \cdot \ys_i\kt = 0, & \forall i \label{eq:kkt:cs_is}
\end{align}
\end{subequations}

\item \textbf{Primal feasibility}
\vspace{-8pt}
\begin{align}
(\ref{eq:inner:c_ded}), (\ref{eq:inner:c_sh}), (\ref{eq:inner:c_pos}) \notag
\end{align}

\item \textbf{Dual feasibility}
\vspace{-8pt}
\begin{align}
(\ref{eq:dual_lambda}), (\ref{eq:dual_nu}), (\ref{eq:dual_gamma}), (\ref{eq:dual_sigma}) \notag
\end{align}
\end{itemize}

In the following, we present a \emph{two-stage water-filling} method for solving the KKT conditions above.

\textbf{Stage one (Dedicated water-filling per slice).}
Each slice~$s$ receives its dedicated budget $\xd_s$, which is allocated among its UEs through a per-slice water-filling procedure. The corresponding water level $\beta_s^\ast\kt$ is defined as the solution to
\begin{equation}
    \sum_{i \in \mathcal{U}_s} \max \left\{ \frac{1}{\beta_s\kt} - \frac{1}{\eta_i\kt},\, 0 \right\} = \xd_s.
\end{equation}
This equation is strictly monotone in $\beta_s\kt$ and can be solved efficiently (e.g., via bisection). The optimal dedicated allocation then becomes
\begin{equation}
    \yd_i\kt = \max \left\{ \frac{1}{\beta_s^{\ast}\kt} - \frac{1}{\eta_i\kt},\, 0 \right\}.
\end{equation}
Conceptually, this stage allocates dedicated PRBs within each slice such that UEs with better channel conditions receive more PRBs, while weaker UEs still obtain a positive share due to the concave log-utility. This ensures proportional fairness under the slice’s dedicated budget.

\begin{figure*}[t]
    \centering
    \includegraphics[width=0.99\linewidth]{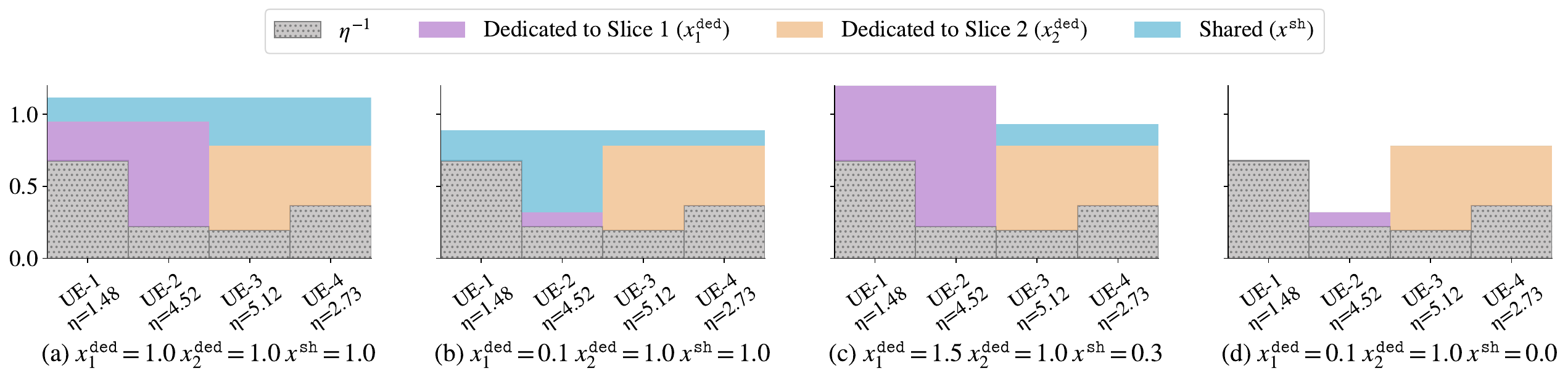}
    \caption{Water-filling interpretation of the inner loop problem. Subplots (a)–(d) show allocations under different dedicated and shared budgets for the same channel conditions across $4$ UEs (first two belonging to slice $1$, last two to slice $2$). Each stacked bar depicts the baseline inverse channel gain, followed by the allocated dedicated resources, and finally the shared resources per UE.}
    \label{fig:waterfilling}
\end{figure*}

\textbf{Stage two (Shared water-filling across slices).}
Once dedicated allocations are set, the shared budget $\xs$ is distributed across all UEs through a global water-filling step. The shared water level $\nu^\ast\kt$ solves
\begin{equation}
    \sum_{i \in \mathcal{U}} 
    \max \left\{ \frac{1}{\nu\kt} - \frac{1}{\eta_i\kt} - \yd_i\kt,\, 0 \right\} = x^{\mathrm{sh}} .
\end{equation}

The resulting shared allocation is
\begin{equation}
    \ys_i\kt = \max \left\{ \frac{1}{\nu^{\ast}\kt} - \frac{1}{\eta_i\kt} - \yd_i\kt,\, 0 \right\} .
\end{equation}

This second stage effectively levels up UE allocations across slices: UEs facing temporarily higher demand draw additional resources from the shared pool, while those under lighter load remain satisfied with their dedicated allocation.

Note that setting the dual variables as $\lambda_s^\ast\kt = \max\{\beta_s^\ast\kt, \nu^\ast\kt\}$ for all $s$ ensures that $(\lambda_s^\ast\kt, \nu^\ast\kt)$ satisfies the KKT conditions, confirming the optimality of the two-stage allocation.

Figure~\ref{fig:waterfilling} illustrates this two-stage water-filling process under different dedicated and shared allocations, highlighting the trade-off between strict isolation of dedicated resources and the statistical multiplexing of shared resources. When both slices receive balanced dedicated budgets and a sufficiently large shared pool (subplot~(a)), resources are leveled both within and across slices, resulting in efficient utilization through multiplexing. Reducing a slice’s dedicated budget while maintaining a shared pool (subplot~(b)) shows how the shared stage compensates for asymmetric slice allocations, effectively redistributing resources to meet fairness objectives. Conversely, increasing dedicated budgets while limiting the shared pool (subplot~(c)) shifts most of the leveling to within-slice allocations, reducing the opportunity for cross-slice multiplexing gains. Finally, removing the shared pool entirely (subplot~(d)) isolates slices completely, eliminating multiplexing benefits and requiring higher total provisioning to achieve the same service guarantees.

\subsection{Solving the Bi-level Optimization Problem}
We convert the bi-level problem, formulated in Section~\ref{sec:outer}, into a single-level one by embedding the inner loop KKT conditions. The inner loop problem’s convexity ensures that this reformulation remains globally optimal. However, solvers such as \texttt{Gurobi} and \texttt{CPLEX} do not support conditions of the form~(\ref{eq:kkt:stat_d})–(\ref{eq:kkt:stat_s}), where decision variables appear in the denominator. To address this, we reformulate these conditions as a set of quadratic constraints. We first observe that
\begin{equation*}
\begin{aligned}
\lambda_{s}\kt &= \frac{\eta_i\kt}{1 + \eta_i\kt \left( \yd_i\kt + \ys_i\kt \right)}  + \gamma_i\kt \\
&\ge  \frac{\eta_i\kt}{1 + \eta_i\kt \left( \yd_i\kt + \ys_i\kt \right)} > 0,
\end{aligned}
\end{equation*}

and define $\omega_{s}\kt := \frac{1}{\lambda_{s}(k,t)}$. Then we claim the set of constraints related to $\yd_i\kt$, namely (\ref{eq:kkt:stat_d}), (\ref{eq:inner:c_ded}), (\ref{eq:kkt:cs_d}), (\ref{eq:kkt:cs_id}), (\ref{eq:dual_lambda}), and (\ref{eq:dual_sigma}) can be transformed to
\begin{subequations}
\begin{align}
&\sum_{i \in \mathcal{U}_s} \yd_i\kt = \xd_s, & \forall s \\
&\eta_i\kt \left( \yd_i\kt + \ys_i\kt - \omega_s\kt\right) \notag \\& \quad \ge -1, & \forall i \\
&\yd_i\kt \big[1 + \eta_i\kt \notag \\& \quad \quad \left( \yd_i\kt + \ys_i\kt - \omega_s\kt \right)  \big] = 0, & \forall i \label{eq:transformed_ded_quad}\\
&0 < \omega_s\kt.  & \forall s
\end{align}
\end{subequations}

The proof is provided in Appendix~\ref{app:proof_prop1}.

Similarly, we claim $\nu\kt > 0$, define $\mu\kt := \frac{1}{\nu(k,t)}$, and transform the constraints related to $\ys_i\kt$, i.e., (\ref{eq:kkt:stat_s}), (\ref{eq:inner:c_sh}), (\ref{eq:kkt:cs_s}), (\ref{eq:kkt:cs_is}), (\ref{eq:dual_nu}), and (\ref{eq:dual_gamma}) to the following constraints:
\begin{subequations}
\begin{align}
&\sum_{i \in \mathcal{U}} \ys_i\kt = \xs, \\
&\eta_i\kt \left( \yd_i\kt + \ys_i\kt - \mu\kt\right) \ge -1, & \forall i \\
&\ys_i\kt \big[1 + \eta_i\kt \notag \\& \quad \quad \left( \yd_i\kt + \ys_i\kt - \mu\kt \right)  \big] = 0, & \forall i \label{eq:transformed_sh_quad}\\
&0 < \mu\kt.
\end{align}
\end{subequations}

The new sets of constraints are linear except for (\ref{eq:transformed_ded_quad}) and (\ref{eq:transformed_sh_quad}), which remain quadratic. To make the problem compatible with the solvers, we apply Big-M encoding to linearize these quadratic conditions. For instance, consider constraint~(\ref{eq:transformed_ded_quad}). Introducing a binary variable $z^{\texttt{ded}}_i\kt \in \{0,1\}$ and a sufficiently large positive constant $M$, we obtain the equivalent mixed-integer linear representation:
\begin{subequations}
\begin{align}
&\yd_i\kt \le z^{\texttt{ded}}_i\kt \; M, \\
& 1 + \eta_i\kt \left( \yd_i\kt + \ys_i\kt - \omega_s\kt \right) \notag\\& \quad \quad \le (1 - z^{\texttt{ded}}_i\kt) \; M, \\
& 1 + \eta_i\kt \left( \yd_i\kt + \ys_i\kt - \omega_s\kt \right) \notag\\& \quad \quad \ge - (1 - z^{\texttt{ded}}_i\kt) \; M.
\end{align}    
\end{subequations}

A similar encoding can be applied to constraint~(\ref{eq:transformed_sh_quad}). In this way, the original non-linear complementarity conditions are reformulated as mixed-integer linear constraints, enabling the use of state-of-the-art solvers such as \texttt{Gurobi} or \texttt{CPLEX}.


\section{Performance Evaluation} \label{sec:eval}
We evaluate the resource savings of \textbf{HyRA} under varying SLA definitions, traffic burstiness, UE counts, and slice counts. We quantify these savings relative to two baselines: \textbf{Dedicated-only}, which reserves resources for each slice, and \textbf{Shared-only}, which relies entirely on a common resource pool. Note that the \textbf{Dedicated-only} baseline represents the upper-bound performance of existing works~\cite{zangooei2023flexible, polese2022colo, liu2020constrained, alcaraz2022model, garces2017network, zheng2018statistical, balasingam2024application}, which solve the problem using learning-based methods and do not guarantee optimality. Both baselines are formulated consistently with HyRA to ensure a fair comparison. The underlying formulations are presented in Appendix~\ref{app:ded_only}, \ref{app:shared_only}, and \ref{app:HyRA}.

We set the simulation time window to $T \!=\! 20\,\text{ms}$, consistent with the policy update interval of near-real-time RAN control in O-RAN, which operates with control loop periods above $10\,\text{ms}$. The per-UE delay average is computed over $K \!=\! 20$ samples drawn from independent traffic and channel realizations (see Section~\ref{sec:prob_form}), ensuring statistical stability of short-term dynamics while maintaining tractability. Empirically, we observed that beyond $K \!=\! 20$, the variance of performance metrics becomes negligible. To further reinforce robustness, each configuration is repeated across $10$ independent random seeds, and the reported allocated PRBs correspond to the average across these $10$ seeds.

UEs are evenly distributed across the standard 3GPP's EPA, EVA, and ETU channel models to capture small, medium, and large scale mobility conditions, respectively. Packet sizes and inter-arrival times follow a Pareto distribution with parameter $\alpha\!=\! 1.5$. Empirical studies have reported alpha values within the range $1$–$2$, indicating that $\alpha \!=\! 1.5$ represents the moderate burstiness in practice~\cite{leland2002self}. We study the impact of burstiness level in Section~\ref{sec:burstiness}. Note that across all experiments in this section, the SLAs are fully satisfied; hence, we omit reporting realized SLA for brevity. 

The complete source code and experimental setup will be made publicly available upon paper acceptance. 

\begin{figure}[t]
    \centering
    \includegraphics[width=0.99\linewidth]{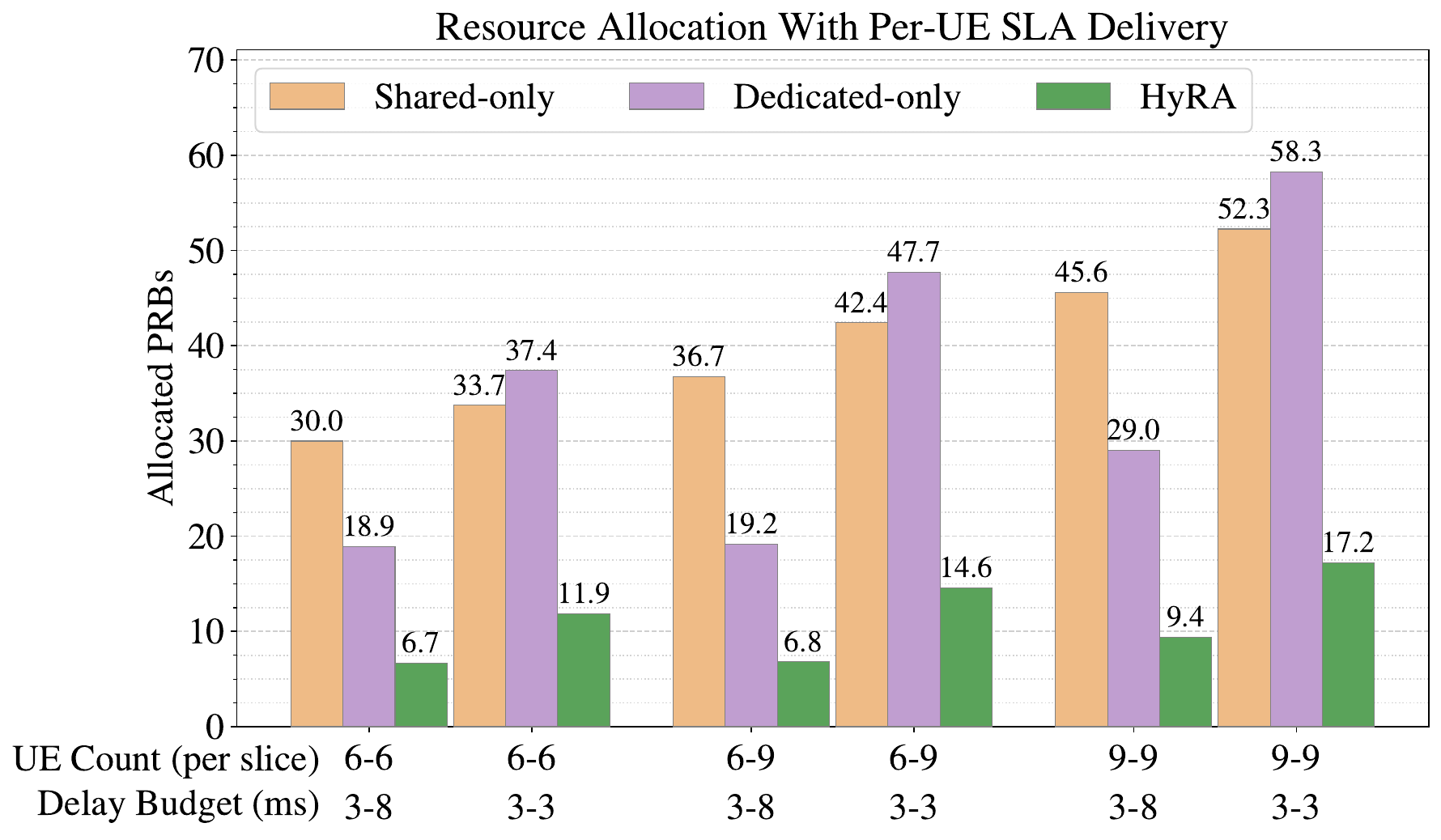}
    \caption{Resource allocation for two network slices to deliver the desired SLA \emph{per UE}. Each group of bars corresponds to a combination of UE counts and delay budgets of the two slices.}
    \label{fig:demand-ue}
\end{figure}

\subsection{Performance Under Varying Resource Demand} \label{sec:demand}
To evaluate HyRA's performance across varying resource-demand scenarios, we vary two key parameters: the number of UEs per slice and the delay budget. In particular, we evaluate cases with ${6}$ and ${9}$ UEs per slice and delay budgets of $3$ and $8$~ms, representing latency-critical and moderate-latency service classes, respectively. The results are presented in Figure~\ref{fig:demand-ue}. The general trend shows that HyRA consistently allocates fewer total resources than the baselines, demonstrating superior resource efficiency through adaptive sharing and dedication mechanisms. On average, HyRA reduces the number of allocated PRBs by approximately $65$–$75\%$ compared to the average of the two baselines. The figure shows that HyRA achieves greater relative savings when service requirements are heterogeneous and the slices have different delay budgets. This indicates its ability to adaptively exploit shared capacity while retaining minimal dedicated resources.

An additional observation concerns the trade-off between the Shared-only and Dedicated-only baselines. In scenarios with heterogeneous service requirements (e.g., delay budgets of $3$-$8$ms), the Dedicated-only approach consumes fewer PRBs, as dedicated provisioning avoids cross-slice interference in the inner loop scheduler and prioritizes latency-sensitive flows. Conversely, under homogeneous service requirements with delay budgets of $3$-$3$ms, the Shared-only scheme becomes more resource-efficient, benefiting from full multiplexing gains. This trade-off underscores the importance of adaptive hybrid strategies such as HyRA, which can dynamically interpolate between these two extremes depending on slice characteristics and demand profiles.

\subsection{Performance Under Slice-aggregated SLA Delivery} \label{sec:aggr_sla}

\begin{figure}[t]
    \centering
    \includegraphics[width=0.99\linewidth]{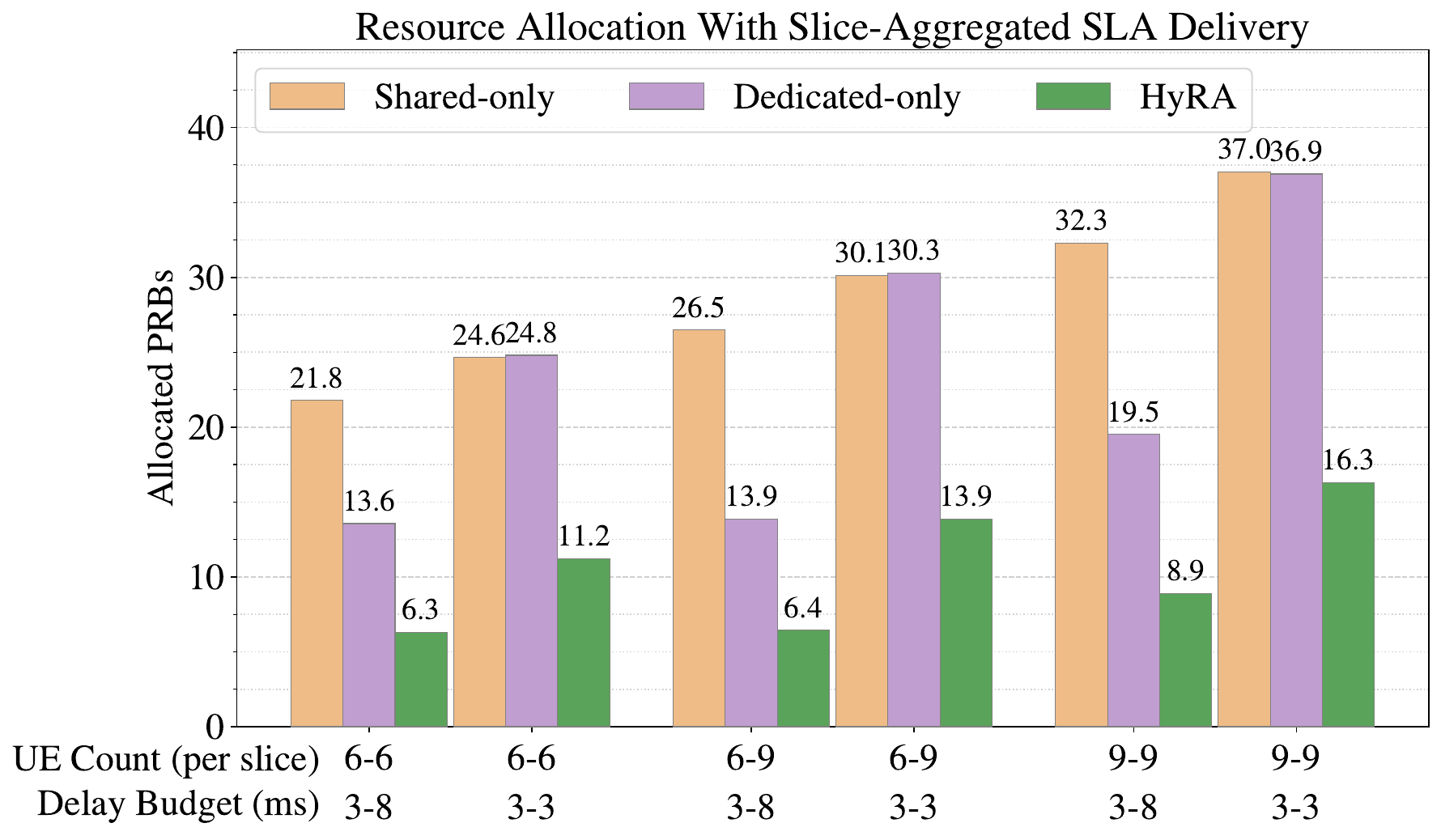}
    \caption{Resource allocation for two network slices to deliver the desired SLA \emph{aggregated over slice UEs}. Each group of bars corresponds to a combination of UE counts and delay budgets of the two slices.}
    \label{fig:demand-slice}
\end{figure}

The constraint~\eqref{eq:HyRA:sla} in the problem formulation enforces stringent \emph{per-UE} SLA satisfaction. However, one might argue that this is not a realistic expectation and define the SLA on the \emph{slice-aggregated} delay, where compliance is measured as the average delay across UEs within a slice. In this case, the constraint is 
\begin{equation}
\frac{\sum_{i \,:\, s(i) = s} \hat{d}_i}{\sum_{i \,:\, s(i) = s} 1} \leq D_{s}.
\end{equation}

To assess the impact of this alternative SLA interpretation, we re-evaluate HyRA's performance against the baselines under the same experimental conditions as in Section~\ref{sec:demand}, varying only the SLA aggregation level. 

The resulting performance comparison is reported in Figure~\ref{fig:demand-slice}. When the SLA is aggregated at the slice level, the performance trends across algorithms remain consistent with those observed under per-UE SLA enforcement, but with generally lower overall resource consumption across all schemes. This reduction stems from the relaxation of individual delay constraints; while the per-UE formulation requires each UE to meet the target delay independently, the slice-level formulation only enforces compliance at the aggregate average level, allowing greater flexibility in resource allocation. 

The results show that HyRA continues to outperform baselines under this relaxed SLA definition. The average HyRA PRB savings relative to the baselines' average remain substantial, ranging between $55\%$ and $65\%$, compared to the $65$-$75\%$ observed in the per-UE SLA definition. The smaller margin can be explained by the fact that the baselines now leverage inter-UE averaging effects within each slice to partially offset inefficiencies. Overall, these results confirm that HyRA's relative advantage persists regardless of the SLA definition.

\subsection{Performance Under Varying Slice Count} \label{sec:scale}

To evaluate \textbf{HyRA}'s performance under different levels of service diversity and scale, we vary the number of slices while keeping the number of UEs per slice at $6$. Each slice corresponds to a distinct service class characterized by a specific delay budget. The system is progressively extended from two to five slices, with delay budgets from $3$ms to $11$ms.

Figure~\ref{fig:scale} shows the average number of allocated PRBs. Across all configurations, HyRA consistently achieves the lowest resource consumption, reducing PRB usage by $60$--$70\%$ relative to the average of the baselines, with savings becoming more pronounced as the number of slices increases. This trend reflects HyRA’s capacity to exploit multiplexing gains across slices with diverse latency targets.

A comparison between the baselines highlights the benefits of hybrid operation. While the \textbf{Shared-only} scheme benefits from statistical multiplexing, its efficiency deteriorates when multiple delay budgets coexist, as contention between latency-sensitive and relaxed flows increases. Conversely, the \textbf{Dedicated-only} approach maintains isolation but fails to capitalize on temporal variations in load and channel conditions. By dynamically blending both strategies, HyRA delivers substantial savings even under five concurrent slices, demonstrating robustness in heterogeneous network scenarios.

\begin{figure}[t]
    \centering
    \includegraphics[width=0.99\linewidth]{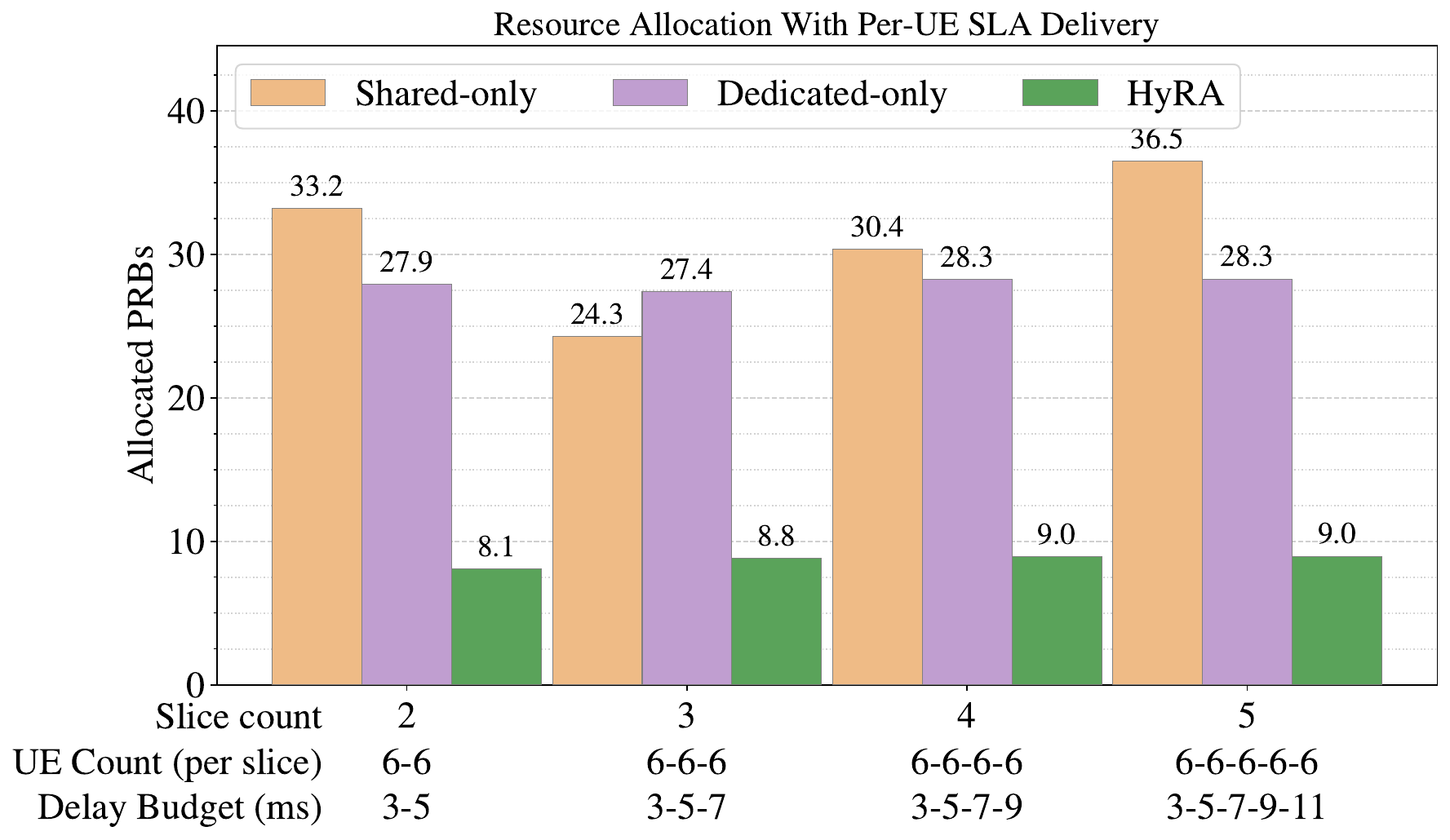}
    \caption{Impact of slice scaling on resource allocation under different delay budgets. Each group of bars corresponds to a combination of slice count and the corresponding delay budgets.}
    \label{fig:scale}
\end{figure}

\subsection{Performance Under Varying Burstiness Level} \label{sec:burstiness}

We evaluate how varying the level of traffic burstiness, defined by the Pareto tail parameter $\alpha$ introduced in Section~\ref{sec:traffic_model}, affects the performance of \textbf{HyRA}. Recall that smaller values of $\alpha$ correspond to heavier-tailed, more bursty traffic, while larger values represent smoother load profiles.

We consider two network slices, each serving six UEs, under two SLA configurations: $3$-$3$~ms and $3$-$8$~ms delay budgets. For both slices, packet sizes and inter-arrival times follow a Pareto distribution, with $\alpha$ equal to $1.95$ (moderate burstiness) and $1.05$ (high burstiness). These values span the empirically observed range for real mobile and Internet traffic~\cite{taqqu1997proof}. Figure~\ref{fig:burstiness} presents the average number of allocated PRBs across the three schemes.

The results demonstrate that \textbf{HyRA} continues to outperform baselines across all burstiness levels, and its relative advantage becomes more pronounced as $\alpha$ decreases. When $\alpha$ is reduced from $1.95$ to $1.05$, the relative PRB savings of HyRA with respect to the mean of the baselines increase from approximately $50\%$ to $62\%$ in the homogeneous-delay case $3$-$3$~ms, and from $61\%$ to $71\%$ in the heterogeneous-delay case $3$-$8$~ms. This trend indicates that HyRA’s hybrid allocation strategy better accommodates bursty traffic, in which instantaneous demand can fluctuate sharply over short timescales. Overall, these results confirm that HyRA not only provides substantial savings under stationary load but also scales resources effectively in environments with substantial temporal traffic variability.

\begin{figure}[t]
    \centering
    \includegraphics[width=0.99\linewidth]{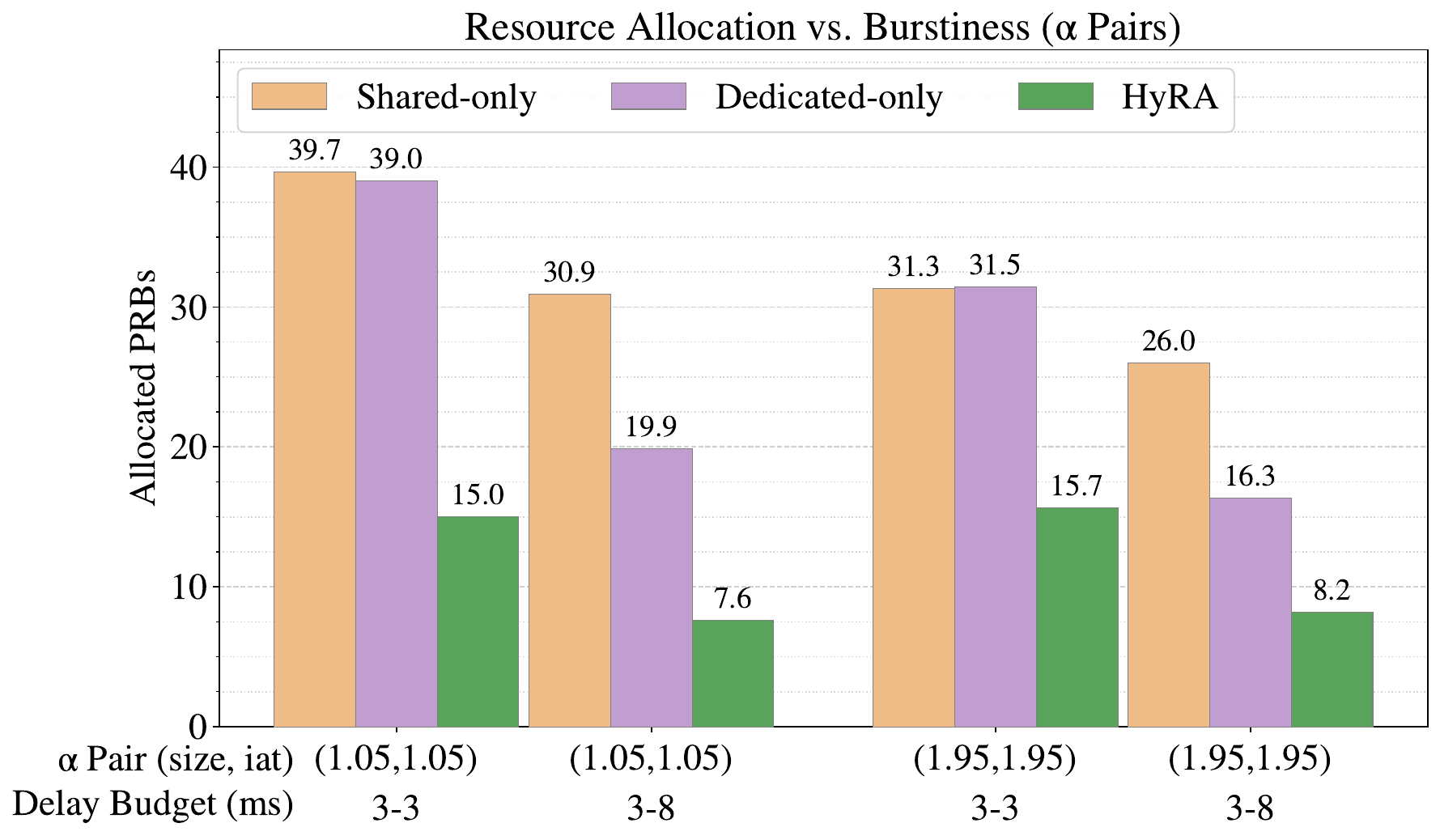}
    \caption{Impact of traffic burstiness on resource allocation under different delay budgets. Lower $\alpha$ indicates higher burstiness.}
    \label{fig:burstiness}
\end{figure}

\section{Discussion and Conclusion}
While the evaluation section demonstrates substantial performance gains with \textbf{HyRA}, it is important to acknowledge that the current optimization framework is not directly applicable to real-world deployment. The execution times for solving the problem optimally observed in our experiments range from a few seconds for the scenarios in Sections~\ref{sec:demand},~\ref{sec:burstiness}, and~\ref{sec:aggr_sla} to tens of seconds for the cases in Section~\ref{sec:scale} (detailed in Appendix~\ref{app:time}), which exceed the one-second time budget of the near-real-time control loop in O-RAN. This computational overhead arises from the exact optimization formulation, which jointly captures inter-slice coupling, stochastic traffic behavior, and SLA constraints.

Nevertheless, these results serve a crucial purpose: they reveal the untapped potential of jointly leveraging shared and dedicated resource pools across network slices by solving the RAN slicing problem optimally. The obtained solutions establish the upper bound of achievable resource efficiency and thus provide a valuable benchmark for designing practical algorithms. This insight aligns with the broader literature, where heuristic and learning-augmented methods~\cite{zangooei2023flexible, polese2022colo, liu2020constrained, liu2021onslicing, alcaraz2022model, garces2017network, zheng2018statistical} are employed to trade off optimality for computational tractability.

Building on the findings of this study, we plan to develop a real-time, learning-based framework that inherits HyRA's hybrid structure. Such a framework will aim to retain the performance advantages of shared–dedicated coordination while operating within the near-real-time constraints of O-RAN control loops.

\newpage

\section*{Acknowledgment}\label{sec:ack}
This work was supported in part by Rogers Communications Canada Inc., NSERC Alliance, and the Ontario Research Fund – Research Excellence program (Project\# ORF-RE012-051) from the Province of Ontario. The views expressed herein are those of the authors and do not necessarily reflect those of the Province.

\bibliographystyle{IEEEtran}
\bibliography{references}

\newpage
\newpage

\appendix

\subsection{Proof of Constraint Transformation} \label{app:proof_prop1}

For given $\{\ys_i\kt \ge 0\}_{\forall i\in \calu}$, define a variable $X = \{\yd_i\kt, \lambda_{s}\kt, \gamma_i\kt \}_{\forall s, \forall i}$, and define the feasible set of $X$ as $\mathcal{X}$ that is specified by constraints:
\begin{subequations}
\label{app:setA}
\begin{align}
&\frac{\eta_i\kt}{1 + \eta_i\kt \left( \yd_i\kt + \ys_i\kt \right)} \notag\\ &\quad= \lambda_{s}\kt - \gamma_i\kt, & \forall s, i \in \mathcal{U}_s \label{app:setA:1} \\
&\sum_{i \in \mathcal{U}_s} \yd_i\kt \leq \xd_s, & \forall s \label{app:setA:2} \\
&\lambda_s\kt \left( \xd_s - \sum_{i \in \mathcal{U}_s} \yd_i\kt \right) = 0, & \forall s \label{app:setA:3}\\
&\gamma_i\kt \, \yd_i\kt = 0. & \forall i \label{app:setA:4}\\
&\gamma_i\kt \ge 0 & \forall i. \label{app:setA:5}
\end{align}    
\end{subequations}

Define another variable $Y = \{\yd_i\kt, \omega_{s}\kt \}_{\forall s, \forall i}$ and its feasible set $\mathcal{Y}$ that is characterized by constraints
\begin{subequations}
\label{app:setB}
\begin{align}
&\sum_{i \in \mathcal{U}_s} \yd_i\kt = \xd_s,  \forall s \label{app:setB:1} \\
&\eta_i\kt \left( \yd_i\kt + \ys_i\kt - \omega_s\kt \right) \ge -1,  \forall i \label{app:setB:2} \\
&\yd_i\kt \Big[ 1 + \eta_i\kt \notag\\& \quad \big( \yd_i\kt + \ys_i\kt - \omega_s\kt \big) \Big] = 0,  \forall i \label{app:setB:3} \\
&\omega_s\kt > 0, \forall s. \label{app:setB:4}
\end{align}
\end{subequations}
We claim the two sets of constraints are equivalent.
\begin{pro}
Given $\{\ys_i\kt\ge 0\}_{i\in \calu}$, constraints~\eqref{app:setA} is equivalent to constraints~\eqref{app:setB}. 
\end{pro}

\textbf{Proof.}
First, if $X = \{\yd_i\kt, \lambda_{s}\kt, \gamma_i\kt \}_{\forall s, \forall i} \in\mathcal{X}$, and we define 
\begin{align*}
    \omega_s\kt = \frac{1}{\lambda_s(k,t)},
\end{align*}
then we can show that $Y = \{\yd_i\kt, \omega_{s}\kt \}_{\forall s, \forall i} \in \mathcal{Y}$.

Recall that 
\begin{align*}
 \lambda_{s}\kt &= \frac{\eta_i\kt}{1 + \eta_i\kt \left( \yd_i\kt + \ys_i\kt \right)}  + \gamma_i\kt \\
&\ge  \frac{\eta_i\kt}{1 + \eta_i\kt \left( \yd_i\kt + \ys_i\kt \right)} > 0. 
\end{align*}
Then $\omega_s\kt > 0$, and the complementary slackness constraint~\eqref{app:setA:3} gives $\sum_{i \in \mathcal{U}_s} \yd_i\kt = \xd_s$. Thus, constraints~\eqref{app:setB:4} and~\eqref{app:setB:1} are satisfied.

Substituting $\omega_s\kt = \frac{1}{\lambda_s \kt }$ into constraint~\eqref{app:setA:1} gives
\begin{align*}
    \gamma_i\kt = 
    &\frac{1 + \eta_i\kt(\yd_i\kt + \ys_i\kt -  \omega_s\kt)}
     {\omega_s\kt(1 + \eta_i\kt(\yd_i\kt + \ys_i\kt))}.
\end{align*}
Since $\gamma_i\kt \ge 0$, and $\gamma_i\kt \, \yd_i\kt = 0$, we can have constraints~\eqref{app:setB:2} and~\eqref{app:setB:3}. Thus, $Y$ is a feasible solution to $\mathcal{Y}$.

Next, suppose $Y = \{\yd_i\kt, \omega_{s}\kt \}_{\forall s, \forall i} \in \mathcal{Y}$, and define 
\begin{align*}
    \lambda_s \kt &= \frac{1}{\omega_{s}\kt},\\
    \gamma_i\kt
&= \lambda_s\kt
- \frac{\eta_i\kt}{1 + \eta_i\kt(\yd_i\kt + \ys_i\kt)}.
\end{align*}
We prove that $X = \{\yd_i\kt, \lambda_{s}\kt, \gamma_i\kt \}_{\forall s, \forall i} \in \mathcal{X}$.

It is straightforward to verify that $X$ satisfy constraints~\eqref{app:setA:1}-\eqref{app:setA:3}. Then substituting $\lambda_s \kt = \frac{1}{\omega_{s}\kt}$ into $\gamma_i\kt$ gives
\begin{align*}
    \gamma_i\kt = 
    &\frac{1 + \eta_i\kt(\yd_i\kt + \ys_i\kt -  \omega_s\kt)}
     {\omega_s\kt(1 + \eta_i\kt(\yd_i\kt + \ys_i\kt))}.
\end{align*}
Then constraints~\eqref{app:setB:2} and~\eqref{app:setB:3} directly indicate that constraints~\eqref{app:setA:5} and~\eqref{app:setA:4}, respectively.

Combining the above two cases completes the proof. $\qed$

\newpage
\subsection{Runtime Performance} \label{app:time}
For all optimization experiments, we used Gurobi running on an Intel(R) Xeon(R) Silver 4314 CPU at 2.40 GHz. Figure~\ref{fig:time} reports the solver runtime for the evaluation scenario with two slices, a horizon of $T=20$, and $K=20$ samples, consistent with the setup described earlier in the evaluation section. Across these configurations, the runtime ranges from approximately $5$ to $30$ seconds. As expected, HyRA consistently exhibits the largest runtime, since its problem formulation includes substantially more variables than the Shared-only and Dedicated-only baselines—capturing both dedicated resource allocations and shared resource interactions. 

As mentioned earlier, shorter runtime is not the primary focus of this work. Our goal is to reveal the untapped potential of jointly leveraging shared and dedicated resource pools by solving the RAN slicing problem optimally. While such an optimal solution provides important insights into the theoretical performance limits, it is inherently too computationally expensive for real-time network operation.

\begin{figure}[h]
    \centering
    \includegraphics[width=0.99\linewidth]{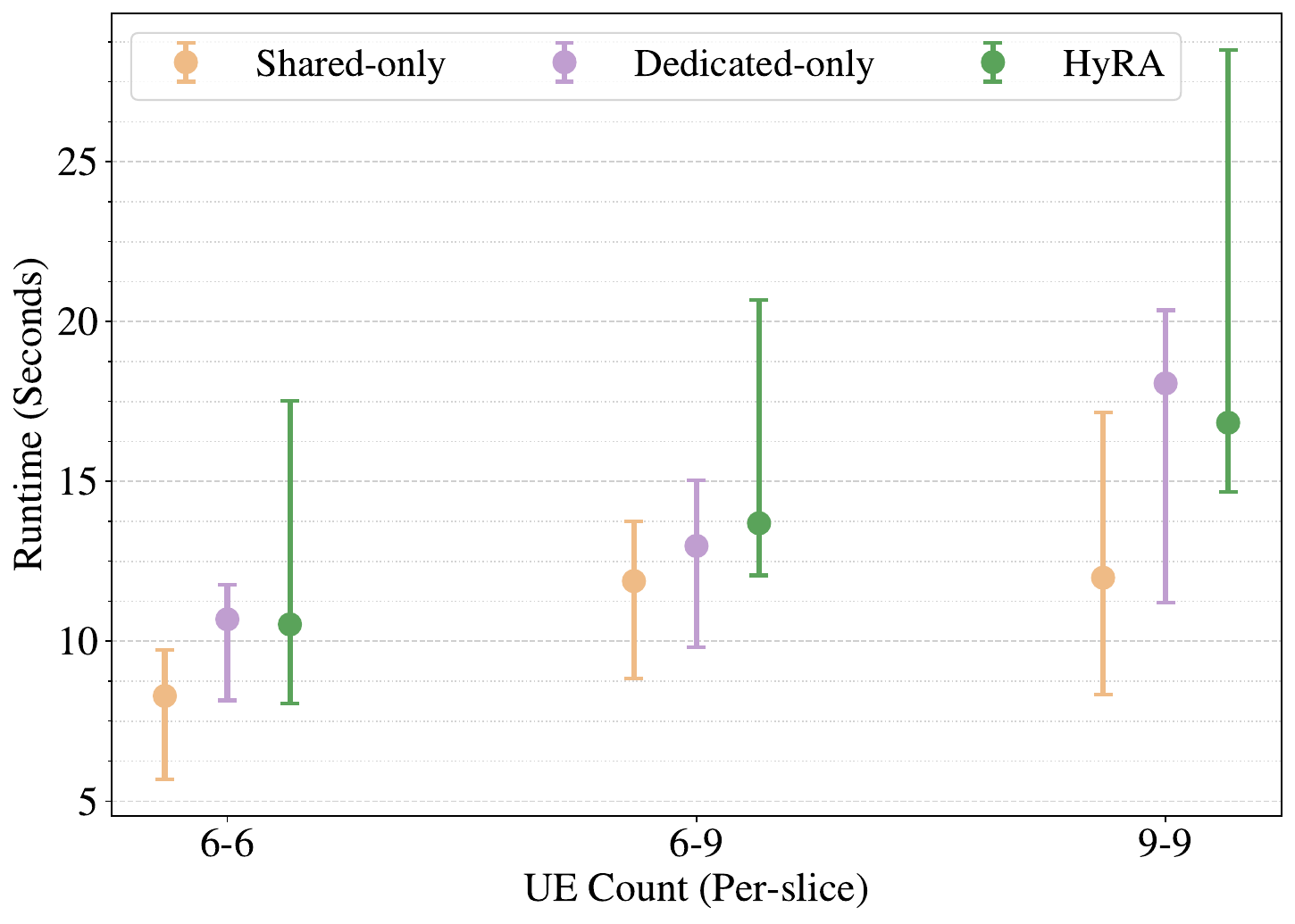}
    \caption{Runtime performance of the three resource allocation algorithms across different per-slice UE configurations. Each point shows the mode runtime over all experiment runs, and the error bars denote the interquartile range (25th–75th percentiles).}
    \label{fig:time}
\end{figure}

\newpage

\subsection{Dedicated-only} \label{app:ded_only}
For comparison, we provide the complete formulation of the \textbf{Dedicated-Only} baseline, in which all resources are pre-assigned to individual slices, with no shared pool. This model represents the traditional static slicing approach and serves as a benchmark for evaluating the performance of HyRA.

\begin{subequations}
\begin{align}
\min & \;\sum_s \xd_s \\&
S_i\kt \leq N_{\texttt{sym}} \eta_i\kt \yd_i\kt, & \forall i, k, t\\&
Q_i\kt + A_i\kt - S_i\kt \leq Q_i(k, t+1), & \forall i, k, t \\&
\frac{1}{K} \sum_{k=1}^{K} 
\frac{
\sum_{t=1}^{T} Q_i\kt
}{
\sum_{t=1}^{T} A_i\kt
}
\leq D_{s,i}, & \forall i \\&
\sum_{i \in \mathcal{U}_s} \yd_i\kt = \xd_s, & \forall s, k, t \\& 
\eta_i\kt (\yd_i\kt - \omega_{s(i)}\kt) \geq -1, & \forall i, k, t \\&
z^{\texttt{ded}}_i\kt \in\{0,1\}, & \forall i, k, t \\& 
\yd_i\kt \le M z^{\texttt{ded}}_i\kt, & \forall i, k, t \\&
1 + \eta_{i}\kt \left(\yd_i\kt - \omega_{s(i)}\kt\right) \notag\\&\quad \le M (1-z^{\texttt{ded}}_i\kt), & \forall i, k, t \\&
- 1 - \eta_{i}\kt\left(\yd_i\kt - \omega_{s(i)}\kt\right) \notag\\&\quad \le M (1-z^{\texttt{ded}}_i\kt), & \forall i, k, t \\&
-1 - \eta_{i}\kt\left(\yd_i\kt + \ys_i\kt - \mu\kt\right) \notag\\&\quad \le M (1-z^{\texttt{ded}}_i\kt), & \forall i, k, t \\&
0 \leq \yd_i\kt, & \forall i, k, t \\&
0 < \omega_s\kt. & \forall s, k, t
\end{align}    
\end{subequations}

\newpage
\subsection{Shared-only} \label{app:shared_only}
This subsection presents the complete formulation of the \textbf{Shared-Only} baseline, where all slices draw their resources from a common pool without any dedicated allocation. This setup highlights the behavior of fully multiplexed resource sharing in contrast to the hybrid HyRA design.

\begin{subequations}
\begin{align}
\min & \; \xs \\&
S_i\kt \leq N_{\texttt{sym}} \eta_i\kt \ys_i\kt, & \forall i, k, t\\&
Q_i\kt + A_i\kt - S_i\kt \leq Q_i(k, t+1), & \forall i, k, t \\&
\frac{1}{K} \sum_{k=1}^{K} 
\frac{
\sum_{t=1}^{T} Q_i\kt
}{
\sum_{t=1}^{T} A_i\kt
}
\leq D_{s,i}, & \forall i \\&
\sum_{i \in \mathcal{U}} \ys_i\kt = \xs, & \forall k, t \\&
\eta_i \kt (\ys_i\kt - \mu\kt) \geq -1 , &  \forall i, k, t \\&
z^{\texttt{sh}}_i\kt \in\{0,1\}, & \forall i, k, t \\&
\ys_i\kt \le M z^{\texttt{sh}}_i\kt, & \forall i, k, t \\&
1 + \eta_{i}\kt\left(\ys_i\kt - \mu\kt\right) \notag\\&\quad \le M (1-z^{\texttt{sh}}_i\kt), & \forall i, k, t \\&
-1 - \eta_{i}\kt\left(\ys_i\kt - \mu\kt\right) \notag\\&\quad \le M (1-z^{\texttt{sh}}_i\kt), & \forall i, k, t \\&
0 \leq \ys_i\kt, & \forall i, k, t \\&
0 < \mu\kt. & \forall k, t
\end{align}    
\end{subequations}

\newpage
\subsection{HyRA} \label{app:HyRA}
This subsection provides the complete MIP formulation of the proposed \textbf{HyRA} framework. 

\begin{subequations}
\begin{align}
\min & \; \xs + \sum_s \xd_s \\&
S_i\kt \leq N_{\texttt{sym}} \eta_i\kt \left( \yd_i\kt + \ys_i\kt \right), & \forall i, k, t\\&
Q_i\kt + A_i\kt - S_i\kt \leq Q_i(k, t+1), & \forall i, k, t \\&
\frac{1}{K} \sum_{k=1}^{K} 
\frac{
\sum_{t=1}^{T} Q_i\kt
}{
\sum_{t=1}^{T} A_i\kt
}
\leq D_{s,i}, & \forall i \\&
\sum_{i \in \mathcal{U}_s} \yd_i\kt = \xd_s, & \forall s, k, t \\& 
\sum_{i \in \mathcal{U}} \ys_i\kt = \xs, & \forall k, t \\&
\eta_i\kt (\yd_i\kt + \ys_i\kt - \omega_{s(i)}\kt) \geq -1, & \forall i, k, t \\&
\eta_i \kt (\yd_i\kt + \ys_i\kt - \mu\kt) \geq -1 , &  \forall i, k, t \\&
z^{\texttt{ded}}_i\kt \in\{0,1\}, & \forall i, k, t \\& 
\yd_i\kt \le M z^{\texttt{ded}}_i\kt, & \forall i, k, t \\&
1 + \eta_{i}\kt \left(\yd_i\kt+\ys_i\kt - \omega_{s(i)}\kt\right) \notag\\&\quad \le M (1-z^{\texttt{ded}}_i\kt), & \forall i, k, t \\&
- 1 - \eta_{i}\kt\left(\yd_i\kt+\ys_i\kt - \omega_{s(i)}\kt\right) \notag\\&\quad \le M (1-z^{\texttt{ded}}_i\kt), & \forall i, k, t \\&
z^{\texttt{sh}}_i\kt \in\{0,1\}, & \forall i, k, t \\&
\ys_i\kt \le M z^{\texttt{sh}}_i\kt, & \forall i, k, t \\&
1 + \eta_{i}\kt\left(\yd_i\kt + \ys_i\kt - \mu\kt\right) \notag\\&\quad \le M (1-z^{\texttt{sh}}_i\kt), & \forall i, k, t \\&
-1 - \eta_{i}\kt\left(\yd_i\kt + \ys_i\kt - \mu\kt\right) \notag\\&\quad \le M (1-z^{\texttt{sh}}_i\kt), & \forall i, k, t \\&
0 \leq \yd_i\kt, \; \ys_i\kt, & \forall i, k, t \\&
0 < \omega_s\kt, \; \mu\kt. & \forall s, k, t
\end{align}    
\end{subequations}

\end{document}